\documentclass[lettersize,journal]{IEEEtran}

\usepackage{cite}
\usepackage{amsmath,amssymb,amsfonts}

\usepackage{graphicx}
\usepackage{etoolbox}
\usepackage{float}
\usepackage{mathrsfs}
\usepackage{subcaption}
\usepackage{amsthm}
\usepackage{float} 
\newtheorem{assumption}{Assumption}
\newtheorem{remark}{Remark}
\newtheorem{lemma}{Lemma}
\newtheorem{theorem}{Theorem}

\usepackage{amsmath,amsfonts}
\usepackage{algorithmic}
\usepackage{algorithm}
\usepackage{array}
\usepackage[caption=false,font=normalsize,labelfont=sf,textfont=sf]{subfig}
\usepackage{textcomp}
\usepackage{stfloats}
\usepackage{url}
\usepackage{verbatim}
\usepackage{graphicx}
\usepackage{cite}
\begin{document}

\title{Observer-Based Data-Driven Consensus Control for Nonlinear Multi-Agent Systems against DoS and FDI attacks}
\author{Yi Zhang, \IEEEmembership{Student Member, IEEE}, Bin Lei, Mohamadamin Rajabinezhad, \IEEEmembership{Student Member, IEEE}, Caiwen Ding, and Shan Zuo, \IEEEmembership{Member, IEEE}
\thanks{Yi Zhang, Mohamadamin Rajabinezhad, and Shan Zuo are with the
Department of Electrical and Computer Engineering, University of Connecticut, Storrs, CT 06269, USA. Bin Lei and Caiwen Ding are with the
Department of Computer Science and Engineering, University of Minnesota Twin Cities, Minneapolis, MN 55455, USA.
(Emails:yi.2.zhang@uconn.edu;lei00126@umn.edu;mohamadamin.
rajabinezhad@uconn.edu;dingc@umn.edu;shan.zuo@uconn.edu.)}}

\markboth{Journal of \LaTeX\ Class Files,~Vol.~14, No.~8, August~2021}%
{Shell \MakeLowercase{\textit{et al.}}: A Sample Article Using IEEEtran.cls for IEEE Journals}

\IEEEpubid{0000--0000/00\$00.00~\copyright~2021 IEEE}

\maketitle

\begin{abstract}
Existing data-driven control methods generally do not address False Data Injection (FDI) and Denial-of-Service (DoS) attacks simultaneously. This letter introduces a distributed data-driven attack-resilient consensus problem under both FDI and DoS attacks and proposes a data-driven consensus control framework, consisting of a group of comprehensive attack-resilient observers. The proposed group of observers is designed to estimate FDI attacks, external disturbances, and lumped disturbances, combined with a DoS attack compensation mechanism. A rigorous stability analysis of the approach is provided to ensure the boundedness of the distributed neighborhood estimation consensus error. 
The effectiveness of the approach is validated through numerical examples involving both leaderless consensus and leader-follower consensus, demonstrating significantly improved resilient performance compared to existing data-driven control approaches.
\end{abstract}


\section{Introduction}
\label{sec:introduction}
\IEEEPARstart{T}{he} FDI attacks and DoS attacks pose significant threats to Multi-Agent Systems (MAS), impacting critical infrastructures such as power systems, transportation networks, and military operations \cite{weng2023secure,zhang2019new}. These attacks exploit the MAS's heavy reliance on communication networks, potentially leading to severe economic losses or societal disruptions \cite{zhang2019new}. Recent research has increasingly focused on resilient control strategies to counter FDI and DoS attacks, which often depend on precise system models. However, generating such models is challenging, as it requires significant time and effort. This difficulty is compounded by unmodeled dynamics and limited robustness. Consequently, there is a growing shift toward data-driven control approaches that eliminate the need for detailed system modeling. Model-Free Adaptive Control (MFAC) has emerged as a promising solution for managing nonlinear MAS with unknown dynamics \cite{hou2010novel}. This approach addresses complex control challenges without requiring detailed knowledge of system dynamics. For instance, \cite{ma2021distributed} proposes a learning-based distributed MFAC algorithm to address nonlinear MAS under DoS attacks. However, this method only addresses DoS attacks and does not handle the common FDI attacks. Similarly, \cite{zhu2022model} introduces an MFAC design for nonlinear MAS under FDI attacks, but the approach relies on an attack detection mechanism. Given that FDI attacks are often stealthy by design, effective detection remains a significant challenge. 

To address these limitations, this letter proposes an observer-based MFAC approach that simultaneously mitigates the adverse effects of both DoS and FDI attacks. The contributions are as follows. 1) A group of attack-resilient observers is theoretically derived using the discrete-time compact-form dynamic linearization (DCFDL) method to construct the attack compensation signal for the distributed neighborhood error, based on the additional DCFDL model of the nonlinear system. 2) Compared to existing MFAC-based attack-resilient controller, which only addresses DoS attacks or FDI attacks, the proposed group of attack-resilient observers can address  both types of attacks simultaneously without relying on detection mechanisms, thus achieving enhanced attack-resilient performance. 3) A rigorous mathematical proof is provided, certifying that the neighborhood consensus estimation error converges to a bound determined by the upper bounds of the lumped disturbance and parameter estimation error, using the contraction mapping principle. That is, the proposed observer-based MFAC scheme guarantees the bounded stability of the nonlinear MAS.

\noindent\textit{Notations:} In this paper, $\left\|X^{p}\right\|$ and $\left\|X^{p \times q}\right\|$ represent the Euclidean norm and 2-norm, respectively. $\Delta$ is the backward difference operator, that is, $\Delta x(k+1)=x(k+1)-x(k)$. $\otimes$ denotes the Kronecker product. $\bar{b}_{(\cdot)}$ denotes an unknown positive constant representing the bound of the variable $(\cdot)$. $\mathbb{E}(\cdot)$ denotes the expectation of a variable $(\cdot)$. \( \operatorname{sgn} \) denotes the sign function. $\mathbf{1}_N\in\mathbb{R}^N$ is a vector with all entries are one. ${I}$ denotes identity matrix with the appropriate dimensions. For a matrix $A\in \mathbb{R}^{n \times n}$, $A_{i,pq}$ denotes its $p$-th row and $q$-th column element.  $\|\cdot\|_{d}$ denotes an induced matrix norm satisfies $\|A x\| \leq \|A\|_{d}\|x\|$.  $\rho(A)=\mathrm{max}\{|z_1|,\cdots,|z_n|\}$ denotes its spectral radius, with eigenvalue $z_i,i=1,\cdots,n$. A MAS consisting of one leader and $N$ followers is considered, where the interactions among them are represented by a signed digraph $\mathscr{G}=(\mathscr{V}, \mathscr{E}, \mathcal{A})$, where $\mathscr{V}=\{0,1,2, \ldots, N\}$ is the set of vertices, $0$ denotes the leader, and $1\cdots N $ denote the followers. $\mathscr{E} \subset \mathscr{V} \times \mathscr{V}$ denotes the set of edges, and $\mathcal{A}=[a_{ij}]\in\mathbb{R}^{N\times N}$ is the associated adjacency matrix where $a_{ij} \neq 0$ if $(j, i) \in \mathscr{E}$. The neighborhood of the agent $i$ is $\mathcal{N}_i=\{j\in\mathscr{V}:(j,i)\in\mathscr{E}\}$ and the self-edge $(i,i)$ satisfies $(i,i)\notin\mathscr{E}$. The in-degree matrix is defined as $\mathcal{D}=\operatorname{diag}\left(\boldsymbol{d}_i\right)$ with $\boldsymbol{d}_i=\sum\nolimits_{j \neq i}a_{ij}$. The Laplacian matrix ${\mathcal{L}}$ is defined as
$\mathcal{L} = \mathcal{D} - \mathcal{A}$.  where $g_i$ is the pinning gain from the leader to the follower $i$. If the $y_{d}(k)$ is available to agent $i$, it satisfies 
$g_i=
1$,
otherwise, $g_i=0$. The pinning gain matrix $\mathcal{G} = \operatorname{diag}(g_i)$.

\section{Problem Formulation}
\label{sec:Problem Formulation}
Consider the following discrete-time nonlinear system for the $N$ followers
\begin{equation}
\label{eq1}
\begin{aligned}
\boldsymbol{y}_i(k+1)= & f_i\left(\boldsymbol{y}_i(k), \boldsymbol{u}_i(k),\boldsymbol{d}_i(k)\right),\; i\in \mathscr{V}
\end{aligned}
\end{equation}where $\boldsymbol{y}_i(k)\in \mathbb{R}^{n_{\boldsymbol{y}_i}}, \boldsymbol{u}_i(k)\in \mathbb{R}^{n_{\boldsymbol{u}_i}},\boldsymbol{d}_i(k)\in \mathbb{R}^{n_{\boldsymbol{d}_i}}$ are the system output, control input, and external disturbance at time $k$, respectively; $n_{\boldsymbol{y}_i}, n_{\boldsymbol{u}_i},n_{\boldsymbol{d}_i}\in \mathbb{Z}_{+}$ are the unknown dimensions of $\boldsymbol{y}_i(k)$, $\boldsymbol{u}_i(k)$, and $\boldsymbol{d}_i(k)$; and $f_i(\cdot): \mathbb{R}^{n_{\boldsymbol{y}_i}+ n_{\boldsymbol{u}_i}+n_{\boldsymbol{d}_i}} \rightarrow \mathbb{R}^{n_{\boldsymbol{y}_i}}$ is an smooth nonlinear function. Four assumptions are given for the system \eqref{eq1}.
\begin{assumption}
\label{ass:1}
The directed topology graph $\mathscr{G}$ has a directed spanning tree.   
\end{assumption}
\begin{assumption}
\label{ass:2}
The system \eqref{eq1} is generalized Lipschitz, that is, $\|\Delta \boldsymbol{y}_i(k+1)\| \leq \bar{b}_{\|\Delta \boldsymbol{y}_i(k+1)\|/\|\Delta \boldsymbol{u}_i(k)\|}\|\Delta \boldsymbol{u}_i(k)\|$ for any $k$ and $\left\|\Delta \boldsymbol{u}_i(k)\right\| \neq 0,$ where $\bar{b}_{\|\Delta \boldsymbol{y}_i(k+1)\|/\|\Delta \boldsymbol{u}_i(k)\|}$ is a positive constant.
\end{assumption}
\begin{assumption}
\label{ass:3}
The external disturbance is bounded, that is, $\|\boldsymbol{d}_i(k)\| \leq \bar{b}_{\boldsymbol{d}_i(k)},$ where $\bar{b}_{\boldsymbol{d}_i(k)}$ is a positive constant.
\end{assumption}
\begin{assumption}
\label{ass:4}
For any $k$, $\frac{\left\|\Delta \boldsymbol{u}_j(k)\right\|}{\left\|\Delta \boldsymbol{u}_i(k)\right\|}<\bar{b}_{\left\|\Delta \boldsymbol{u}_j(k)\right\|/\left\|\Delta \boldsymbol{u}_i(k)\right\|},$ $i, j \in \mathscr{V}, $ where $\bar{b}_{\left\|\Delta \boldsymbol{u}_j(k)\right\|/\left\|\Delta \boldsymbol{u}_i(k)\right\|}$ is a positive constant.
\end{assumption}

\begin{remark}
\label{rem:1}
These assumptions are standard in the literature. Assumption \ref{ass:1} is a  necessary condition to achieve leader-follower consensus. Assumption \ref{ass:2} indicates limited change of control input $\boldsymbol{u}_i(k)$ will not cause unlimited change of system output $\boldsymbol{y}_i(k+1)$. This means that the system energy is finite, and most practical systems satisfy this assumption. For Assumption \ref{ass:3}, the external disturbances on practical systems are usually bounded. Assumption \ref{ass:4} posits that controllers stabilize the system state, with input variations as energy. Unbounded input ratios may cause energy imbalances, risking instability \cite{zhang2024resilient}.
\end{remark}

The following technical results is needed.
\begin{lemma}[\cite{phanomchoeng2010bounded} Mean Value Theorem for a Vector Function] 
\label{lem:1}
Let the canonical basis of the vectorial space \( \mathbb{R}^p \) for all \( p \geq 1 \) be defined by $E_p = \{ e_p(i) \mid e_p(i) = (0, \dots, 0, \underset{i}{1}, 0, \dots, 0)^{\mathrm{T}}, i = 1, \dots, p \};$
and $f(x): \mathbb{R}^p \to \mathbb{R}^p$ be a function continuous on $[a,b] \in \mathbb{R}^p$ and differentiable on the convex hull of the set $(a,b)$. For $y_1, y_2 \in [a, b]$, there exist $\varrho_{ij,x}^{f,\text{max}}$ and $\varrho_{ij,x}^{f,\text{min}}$ for $i = 1, \dots, p$ and $j = 1, \dots, p$ such that $f(y_2) - f(y_1) 
= \Xi_{f}^{\Delta y}(y_2 - y_1),$ where $\Xi_{f}^{\Delta y}=\sum_{i=1}^p \sum_{j=1}^p e_p(i) e_p^{\mathrm{T}}(j) \phi_{ij,x}^{f,\text{max}} \varrho_{ij,x}^{f,\text{max}}
+\sum_{i=1}^p \sum_{j=1}^p e_p(i) e_p^{\mathrm{T}}(j) \phi_{ij,x}^{f,\text{min}} \varrho_{ij,x}^{f,\text{min}}$; $\varrho_{ij,x}^{f,\text{max}}, \varrho_{ij,x}^{f,\text{min}} \leq 0$, $\varrho_{ij,x}^{f,\text{max}} + \varrho_{ij,x}^{f,\text{min}} = 1$, $\phi_{ij,x}^{f,\text{max}} \geq \max\left(\frac{\partial f_i}{\partial x_j}\right)$, and $\phi_{ij,x}^{f,\text{min}} \leq \min\left(\frac{\partial f_i}{\partial x_j}\right)$,
for all $x \in (a, b)$.
\end{lemma}

\begin{lemma}
\label{thm:1}
Consider the discrete-time nonlinear system \eqref{eq1} satisfying Assumptions \ref{ass:1}-\ref{ass:4}. Then, system \eqref{eq1} can be equivalently transformed as the following DCFDL model
\begin{equation}
\label{eq2}
\Delta \boldsymbol{y}_i(k+1)=\Theta_i(k) \Delta \boldsymbol{u}_i(k)+\boldsymbol{\tau}_i(k)
\end{equation}
where $\Theta_i(k)$ is the pseudo-partitioned Jacobian matrix (PPJM) of $f_i(\cdot)$ in relation to $\boldsymbol{u}_i(k)$, which satisfies $\|\Theta_i(k)\|\leq\bar{b}_{\Theta_i(k)}$, where $\bar{b}_{\Theta_i(k)}$ is a positive constant; and $\boldsymbol{\tau}_i(k)$ is the residual nonlinear lumped disturbance.
\end{lemma}
\begin{proof}
Based on Lemma \ref{lem:1}, by using \eqref{eq1}, one has 
\begin{equation*}
\begin{aligned}
&\Delta \boldsymbol{y}_i(k+1)
\\=&    f_i\left(\boldsymbol{y}_i(k), \boldsymbol{u}_i(k),\boldsymbol{d}_i(k)\right) 
 -f_i\left(\boldsymbol{y}_i(k), \boldsymbol{u}_i(k-1),\boldsymbol{d}_i(k)\right) 
\\&+f_i\left(\boldsymbol{y}_i(k), \boldsymbol{u}_i(k-1),\boldsymbol{d}_i(k)\right)-f_i\left(\boldsymbol{y}_i(k-1),\cdots\right.
\\& \left.\boldsymbol{u}_i(k-
1),\boldsymbol{d}_i(k)\right)
+f_i\left(\boldsymbol{y}_i(k-1), \boldsymbol{u}_i(k-
1),\boldsymbol{d}_i(k)\right)
\\&-f_i\left(\boldsymbol{y}_i(k-1), \boldsymbol{u}_i(k-1),\boldsymbol{d}_i(k-1)\right)
\\=&\Theta_i(k) \Delta \boldsymbol{u}_i(k)+\boldsymbol{\tau}_i(k)
\end{aligned}
\end{equation*}
where $\Theta_i(k)=
\sum_{m=1}^p \sum_{n=1}^p e_p(m) e_p^{\mathrm{T}}(n) \phi_{mn,\boldsymbol{u}_i}^{f_i,\text{max}} \varrho_{mn,\boldsymbol{u}_i}^{f_i,\text{max}}
+\sum_{m=1}^p \sum_{n=1}^p e_p(m) e_p^{\mathrm{T}}(n) \phi_{mn,\boldsymbol{u}_i}^{f_i,\text{min}} \varrho_{mn,\boldsymbol{u}_i}^{f_i,\text{min}}$; and $\boldsymbol{\tau}_i(k)=(
\sum_{m=1}^p \sum_{n=1}^p e_p(m) e_p^{\mathrm{T}}(n) \phi_{mn,\boldsymbol{y}_i}^{f_i,\text{max}} \varrho_{mn,\boldsymbol{y}_i}^{f_i,\text{max}}
+\sum_{m=1}^p \sum_{n=1}^p e_p(m) e_p^{\mathrm{T}}(n) \phi_{mn,\boldsymbol{y}_i}^{f_i,\text{min}} \varrho_{mn,\boldsymbol{y}_i}^{f_i,\text{min}} )\Delta \boldsymbol{y}_i(k)
+(\sum_{m=1}^p \sum_{n=1}^p e_p(m) e_p^{\mathrm{T}}(n) \phi_{mn,\boldsymbol{d}_i}^{f_i,\text{max}} \varrho_{mn,\boldsymbol{d}_i}^{f_i,\text{max}}
+\sum_{m=1}^p \sum_{n=1}^p e_p(m) e_p^{\mathrm{T}}(n) \phi_{mn,\boldsymbol{d}_i}^{f_i,\text{min}} \varrho_{mn,\boldsymbol{d}_i}^{f_i,\text{min}})\Delta \boldsymbol{d}_i(k).$ Based on Assumption \ref{ass:2}, $f_i(\cdot)$ is generalized Lipschitz, that is, there exists $\bar{b}_{\Theta_i}$ such that $\|\Delta \boldsymbol{y}_i(k+1)\|\leq\bar{b}_{\Theta_i} \|\Delta \boldsymbol{u}_i(k)\|.$ Hence, $\|\Theta_i(k)\|\leq\bar{b}_{\Theta_i}.$ This completes the proof.
\end{proof}  

The following objective is considered in this paper.

\noindent\textbf{Objective} (Leader-Follower Consensus Objective)\textbf{.} \textit{Consider the followers’ dynamics \eqref{eq1} and a desired trajectory $y_{0}(k)$, to which only a portion of the followers have access, the leader-follower consensus objective is to ensure that the outputs of the followers satisfy $\lim _{k \rightarrow \infty} y_{i}(k)=y_{0}(k)$.}    

To this end, define the neighborhood consensus error as
\begin{flalign}
\label{eq3}
\xi_i(k)=\sum\nolimits_{j \in N_{i}} a_{i j}\left(y_{i}(k)-\boldsymbol{y}_j(k)\right)+g_i\left(y_{0}(k)-y_{i}(k)\right)
&&\raisetag{1\baselineskip}
\end{flalign}

The global form of \eqref{eq3} is $\xi(k) = (\mathcal{L}+\mathcal{G})\otimes I_N(\bar{y}_0(k)-y(k))$, where $\xi(k)=[\xi_1(k)^{\mathrm{T}},\cdots,\xi_N(k)^{\mathrm{T}}]^{\mathrm{T}}$, $\bar{y}_0(k)=y_0(k)\otimes \boldsymbol{1}_N$, and $y(k)=[y_1(k)^{\mathrm{T}},\cdots,y_N(k)^{\mathrm{T}}]^{\mathrm{T}}$. Hence, when $\lim_{k\to\infty}\xi_i(k)=0$, the objective is achieved.

In this paper, a persistent FDI attack combining with an aperiodic DoS attack is considered. The attacked output signal $\boldsymbol{y}_i^a(k)$ is modeled as $\boldsymbol{y}_i^a(k)=H_i(k)\left(\boldsymbol{y}_i(k)+\boldsymbol{\delta}_i(k,\boldsymbol{y}_i(k))\right)$, where $\boldsymbol{\delta}_i(k,\boldsymbol{y}_i(k))\in\mathbb{R}^{n_{\boldsymbol{y}_i}}$ denotes the FDI attack signal injected into the output signal transmitted by the measurement channel; and $H_i(k)=\operatorname{diag}\left(h_{i,m}(k)\right), m\in\{1,\cdots, n_{\boldsymbol{y}_i}\}$ is the DoS attacks coefficient, which satisfies $h_{i,m}(k)=1, k\in\left[T_{i,m,n-1}^{\text{OFF}}, T_{i,m,n}^{\text{ON}}\right)$; $0, k\in\left[T_{i,m,n}^{\text{ON}}, T_{i,m,n}^{\text{OFF}}\right)$; and $\left[T_{n}^{\mathrm{ON}}, T_{n}^{\mathrm{OFF}}\right)$ represents the $n$th DoS attacks interval. The start and end time instants are defined as $T_{n}^{\mathrm{ON}}$ and $T_{n}^{\mathrm{OFF}}$, respectively.

The following assumptions of the FDI and DoS attacks are considered.
\begin{assumption}
\label{ass:5}
 $\|\boldsymbol{\delta}_{i}(k, \boldsymbol{y}_i(k))\| \leq \bar{b}_{\|\boldsymbol{\delta}_{i}(k, \boldsymbol{y}_i(k))\|},$ where $\bar{b}_{\|\boldsymbol{\delta}_{i}(k, \boldsymbol{y}_i(k))\|}$ is a positive constant.
\end{assumption}
\begin{assumption}
[DoS Attack Frequency]
\label{ass:6}
For any $\left[k_{0}, k\right] \subset$ $[0, \infty)$, there must be positive constants $\kappa_{a}$ and $K_{\infty}$ function $f_{\kappa}\left(k_{0}, k\right)$ such that $n_{a}\left(k_{0}, k\right) \leq \kappa_{a}+f_{a}\left(k_{0}, k\right)$, where $n_{a}\left(k_{0}, k\right)$ is the total number of DoS attacks over $\left[k_{0}, k\right]$.    
\end{assumption}
\begin{assumption}[DoS Attack Duration]
\label{ass:7}
For any $\left[k_{0}, k\right] \subset$ $[0, \infty)$, there must be positive constants $\Xi_a$ and $K_{\infty}$ function $f_{\Xi}\left(k_{0}, k\right)$ such that $\Xi\left(k_{0}, k\right) \leq \Xi_a+f_{\Xi}\left(k_{0}, k\right)$, where $\Xi$ is the total duration time of DoS attacks over $\left[k_{0}, k\right]$.
\end{assumption}
\begin{remark}
\label{rem:2}
Considering that the attacker has knowledge of the security detection mechanism and can eavesdrop on the information disclosed on the channel, and it can use the eavesdropped information in the channel to reconstruct the attack signals. Due to the limited power of the attacker, the FDI attack signal generally has an upper bound, as shown in Assumption \ref{ass:5}. The rationality of the Assumption \ref{ass:6} and \ref{ass:7}
has been extensively discussed in \cite{de2015input}, and thus it is
omitted.
\end{remark}

To achieve the leader-follower consensus objective of the discrete-time nonlinear MAS, the attack-resilient distributed consensus problem is designed as follows.

\noindent\textbf{Problem} (Attack-Resilient Distributed Consensus Problem)\textbf{.}
\textit{Given Assumptions \ref{ass:1}-\ref{ass:7}, the attack-resilient distributed consensus problem is to develop an observer-based data-driven consensus control algorithm such that $\mathbb{E}\left(\left\|\xi_i(k)\right\|\right)$ is bounded, that is, $\lim_{k\to\infty}\mathbb{E}\left(\left\|\xi_i(k)\right\|\right)\leq \bar{b}_{\mathbb{E}\left(\left\|\xi_i(k)\right\|\right)},$ where $\bar{b}_{\mathbb{E}\left(\left\|\xi_i(k)\right\|\right)}$ is a positive constant.}

\section{Control System Design and Stability Analysis}
\label{sec:Control System Design and Stability Analysis}

To address the DoS attack, an attack compensation observer is designed as\begin{equation}
\label{eq4}
\boldsymbol{\chi}_i(k) = H_i(k)\xi_i(k)+(I-H_i(k))\xi_i(k-1)    
\end{equation}
\begin{theorem}
\label{thm:2}
Given Assumptions \ref{ass:1}-\ref{ass:7},  there exists a PPJM $\Phi_{i}(k)$ such that\begin{equation}
\label{eq5}
\Delta \boldsymbol{\chi}_{i}(k+1)=\Phi_{i}(k) \Delta \boldsymbol{u}_i(k)+\boldsymbol{\theta}_i(k)
\end{equation}
where $\boldsymbol{\theta}_i(k)$ is the residual nonlinear lumped disturbance, and $\Phi_{i}(k)$ is bounded such that $\left\|\Phi_{i}(k)\right\| \leq \bar{b}_{\Phi_i}$.    
\end{theorem}

\begin{proof}
\label{prf:1}
Define $\zeta_i(k) = (I-H_i(k))\xi_i(k-1)$ and $\left\{\left(\boldsymbol{y}_j(k),\boldsymbol{u}_j(k)\right)\right\}_{j \in {\mathcal{N}_i}}\equiv\{\left(\boldsymbol{y}_j(k),\boldsymbol{u}_j(k))\mid j \in {\mathcal{N}_i}\right.\}.$ 

Due to the page limit, the calculation is omitted. Substituting \eqref{eq1} and  \eqref{eq3} into \eqref{eq5}, one has $\boldsymbol{\chi}_i(k+1)=\bar{f}_i\left(\boldsymbol{y}_i(k), \boldsymbol{u}_i(k), \{\left(\boldsymbol{y}_j(k),\boldsymbol{u}_j(k)\right)\right)\}_{j \in \mathcal{N}_i}) +\zeta_i(k)$, where $\bar{f}_i(\cdot)$ is a smooth nonlinear function, as $f_i(\cdot)$ is a smooth nonlinear function, and $\bar{f}_i(\cdot)$ is a linear combination of $f_i(\cdot)$. Then based on Lemmas \ref{lem:1} and \ref{thm:1}, following a process similar to the proof of Lemma \ref{thm:1}, one has $\Delta \boldsymbol{\chi}_{i}(k+1)=\Phi_{i}(k) \Delta \boldsymbol{u}_i(k)+\boldsymbol{\theta}_i(k)$, where $\Phi_{i}(k)=
\sum_{m=1}^p \sum_{n=1}^p e_p(m) e_p^{\mathrm{T}}(n) \phi_{mn,\boldsymbol{u}_i}^{\bar{f}_i,\text{max}} \varrho_{mn,\boldsymbol{u}_i}^{\bar{f}_i,\text{max}}
+\sum_{m=1}^p \sum_{n=1}^p e_p(m) e_p^{\mathrm{T}}(n) \phi_{mn,\boldsymbol{u}_i}^{\bar{f}_i,\text{min}} \varrho_{mn,\boldsymbol{u}_i}^{\bar{f}_i,\text{min}}  
+\sum_{m=1}^p \sum_{n=1}^p e_p(m) e_p^{\mathrm{T}}(n) \phi_{mn,\boldsymbol{y}_i}^{\bar{f}_i,\text{max}} \varrho_{mn,\boldsymbol{y}_i}^{\bar{f}_i,\text{max}}\Theta_i(k)
+\sum_{m=1}^p \sum_{n=1}^p e_p(m) e_p^{\mathrm{T}}(n) \phi_{mn,\boldsymbol{y}_i}^{\bar{f}_i,\text{min}} \varrho_{mn,\boldsymbol{y}_i}^{\bar{f}_i,\text{min}}\Theta_i(k)$; and $\boldsymbol{\theta}_i(k) =
\sum_{m=1}^p \sum_{n=1}^p e_p(m) e_p^{\mathrm{T}}(n) \phi_{mn,\boldsymbol{y}_i}^{\bar{f}_i,\text{max}} \varrho_{mn,\boldsymbol{y}_i}^{\bar{f}_i,\text{max}}\tau_i(k)
+\sum_{m=1}^p \sum_{n=1}^p e_p(m) e_p^{\mathrm{T}}(n) \phi_{mn,\boldsymbol{y}_i}^{\bar{f}_i,\text{min}} \varrho_{mn,\boldsymbol{y}_i}^{\bar{f}_i,\text{min}}\tau_i(k)+\sum_{m=1}^p \sum_{n=1}^p e_p(m) e_p^{\mathrm{T}}(n) \phi_{mn,\boldsymbol{y}_j}^{\bar{f}_i,\text{max}} \varrho_{mn,\boldsymbol{y}_j}^{\bar{f}_i,\text{max}}\Delta \boldsymbol{y}_j(k)
+\sum_{m=1}^p \sum_{n=1}^p e_p(m) e_p^{\mathrm{T}}(n) \phi_{mn,\boldsymbol{y}_j}^{\bar{f}_i,\text{min}} \varrho_{mn,\boldsymbol{y}_j}^{\bar{f}_i,\text{min}}\Delta \boldsymbol{y}_j(k)
+\sum_{m=1}^p \sum_{n=1}^p e_p(m) e_p^{\mathrm{T}}(n) \phi_{mn,\boldsymbol{u}_j}^{\bar{f}_i,\text{max}} \varrho_{mn,\boldsymbol{u}_j}^{\bar{f}_i,\text{max}}\Delta \boldsymbol{u}_j(k)
+\sum_{m=1}^p \sum_{n=1}^p e_p(m) e_p^{\mathrm{T}}(n) \phi_{mn,\boldsymbol{u}_j}^{\bar{f}_i,\text{min}} \varrho_{mn,\boldsymbol{u}_j}^{\bar{f}_i,\text{min}}\Delta \boldsymbol{u}_j(k)+\Delta\zeta_i(k).$ 

Since $\boldsymbol{\chi}_i(k+1)$ is a linear combination of unknown smooth nonlinear functions $ {f}_i(\cdot)$ and $ {f}_j(\cdot)$. These functions satisfy Assumption \ref{ass:2}, that is, they are all generalized Lipschitz. Hence, $\boldsymbol{\chi}_i(k+1)$ is also generalized Lipschitz, and consequently, $\Delta\boldsymbol{\chi}_i(k+1)$ is generalized Lipschitz as well. This implies that there exists a positive constant $\bar{b}_{\Phi_i(k)}$ such that $\left\|\Delta\boldsymbol{\chi}_i(k+1)\right\|\leq \bar{b}_{\Phi_i(k)}\left\|\Delta\boldsymbol{u}_i(k)\right\|+\left\|\boldsymbol{\theta}_i(k)\right\|.$ This completes the proof.
\end{proof}

For the DCFDL model in \eqref{eq5}, a disturbance observer should theoretically exist to estimate the lumped disturbance $\boldsymbol{\theta}_i(k)$, as various types of disturbance observers have been successfully applied and proven effective in practice. Therefore, with the design of an external disturbance observer and an FDI attack observer, an effective group of observers is designed as
\begin{align}
&\label{eq6}\begin{aligned}
\hat{\boldsymbol{\chi}}_i(k+1)= &\hat{\boldsymbol{\chi}}_i(k)+\left(\hat{\Phi}_i(k)\Delta \boldsymbol{u}_i(k)
+ \hat{\boldsymbol{\theta}}_i(k)\right)\\&+l_1\hat{\boldsymbol{d}}_i(k)+l_2\hat{\boldsymbol{\delta}}_{i}(k)+l_{3} \tilde{\boldsymbol{\chi}}_i(k)
\end{aligned}
\\\label{eq7}&\hat{\boldsymbol{\theta}}_i(k+1)=\hat{\boldsymbol{\theta}}_i(k)+l_4\hat{\boldsymbol{d}}_i(k)+l_5\hat{\boldsymbol{\delta}}_{i}+l_6 \tilde{\boldsymbol{\chi}}_i(k)
\\\label{eq8}&\hat{\boldsymbol{d}}_i(k+1)=\hat{\boldsymbol{d}}_{i}(k)+l_7\left(\hat{\Phi}_i(k)\Delta \boldsymbol{u}_i(k)+\hat{\boldsymbol{\theta}}_i(k)\right)
\\\label{eq9}&\hat{\boldsymbol{\delta}}_{i}(k+1)=\hat{\boldsymbol{\delta}}_{i}(k)+l_8\left(\hat{\Phi}_i(k)\Delta \boldsymbol{u}_i(k)+\hat{\boldsymbol{\theta}}_i(k)\right)
\end{align}
where $\tilde{\boldsymbol{\chi}}_i(k+1)=\boldsymbol{\chi}_{i}(k+1)-\hat{\boldsymbol{\chi}}_i(k+1)$.

Since the time-varying PPJM $\Phi_{i}(k)$ is difficult to obtain, a performance function is designed as $J_{1}[\Phi_{i}(k)]=\|\Delta \boldsymbol{\chi}_{i}(k)-\hat{\Phi}_i(k-1) \Delta \boldsymbol{u}_i(k-1)\|^{2}
+\rho_i\|\Phi_{i}(k)-\hat{\Phi}_{i}(k-1)\|^{2},$ where $\rho_i>0$ is a weighting factor  constraining the change rate
of PPJMs.
Based on it, $\hat{\Phi}_i(k)$ is designed as
\begin{equation}
\label{eq10}
\begin{aligned}
\hat{\Phi}_i(k)=&\hat{\Phi}_i(k-1) +\frac{\rho_i  \Delta\boldsymbol{\chi}_i(k)\Delta \boldsymbol{u}_i^{\mathrm{T}}(k-1)}{\lambda_i+\|\Delta \boldsymbol{u}_i(k-1)\|^2}
\\&-\frac{\rho_i \hat{\Phi}_i(k-1) \Delta \boldsymbol{u}_i(k-1)\Delta \boldsymbol{u}_i^{\mathrm{T}}(k-1)}{\lambda_i+\|\Delta \boldsymbol{u}_i(k-1)\|^2}    
\end{aligned}    
\end{equation}
where $\rho_i\in(0,1]$ is a step factor.

A resetting mechanism is designed to enhance stability as
\begin{equation}
\label{eq11}
\begin{aligned}
&\hat{\Phi}_{i}(0)=\hat{\Phi}_{i}(1),
\text { if }\|\hat{\Phi}_{i}(k)\|<\epsilon_1, 
\text { or }\|\Delta \boldsymbol{u}_i(k)\|<\epsilon_2, 
\\&\text{ or } \operatorname{sgn}(\hat{\Phi}_{i,pp}(k)) \neq \operatorname{sgn}(\hat{\Phi}_{i,pp}(0)),
\\&\text{ or }\operatorname{sgn}(\hat{\Phi}_{i,pq}(k)) \neq \operatorname{sgn}(\hat{\Phi}_{i,pq}(0)).
\end{aligned}   
\end{equation}
where $\epsilon_1$ and $\epsilon_2$ are small positive constant, and $\hat{\Phi}_{i,pp}(0)$ and $\hat{\Phi}_{i,pq}(0)$ being the initial values of $\hat{\Phi}_{i,pp}(k)$ and $\hat{\Phi}_{i,pq}(k)$, respectively.

To design the distributed data-driven consensus controller, a performance function is designed as $J_{2}[\boldsymbol{u}_i(k)]=\|\hat{\boldsymbol{\chi}}_{i}(k+1)\|^{2}+\mu_i\|\boldsymbol{u}_i(k)-\boldsymbol{u}_i(k-1)\|^{2},$ where $\mu_i>0$ is a weighting factor  constraining the change rate
of PPJMs.
Based on \eqref{eq14}, $\boldsymbol{u}_i(k)$ is designed as
\begin{equation}
\label{eq12}
\begin{aligned}
\boldsymbol{u}_i(k) = &\boldsymbol{u}_i(k-1) 
-\frac{\eta_i\hat{\Phi}_i(k)^{\mathrm{T}}\left(\hat{\boldsymbol{\chi}}_i(k)+l_3 \tilde{\boldsymbol{\chi}}_i(k)\right)}{\mu_i+\|\hat{\Phi}_i(k)\|^2}
\\&-\frac{\eta_i\hat{\Phi}_i(k)^{\mathrm{T}}\left(l_1\hat{\boldsymbol{\delta}}_{i}(k)+l_2\hat{\boldsymbol{d}}_i(k)\right)}{\mu_i+\|\hat{\Phi}_i(k)\|^2}
\end{aligned}
\end{equation}
where $\eta_i\in(0,1]$ is a step factor.

The following two technical Lemmas are needed for the stability analysis.

\begin{lemma}[\cite{bell1965gershgorin} Gersgorin Disk Theorem]
\label{lem:2}
Let $A = [a_{ij}]_{n \times n}$ be a complex matrix and $R_i$ be the sum of the moduli of the off-diagonal elements in the $i$th row. Then, each eigenvalue of $A$ lies in the unions of the circle $
|z - a_{ii}| \leq R_i = \sum_{j=1, j \neq i}^{n} |a_{ij}|, \quad i \in \{1, \dots, n\}.$ 
\end{lemma}
\begin{lemma}[\cite{ortega2000iterative}]
\label{lem:3}
Let $A \in \mathbb{R}^{n \times n}$, for any given $\kappa > 0$, there exists an induced consistent matrix norm such that $\|A\|_v \leq s(A) + \kappa,$ where $s(A)$ is the spectral radius of $A$.
\end{lemma}

The following theorem guarantees the leader-follower consensus control for MAS.

\begin{theorem}
\label{thm:3}
Given Assumptions \ref{ass:1}-\ref{ass:7}, by using the observer-based data-driven consensus control algorithms from Eq. \eqref{eq6} to Eq. \eqref{eq12}, there exist $\mu_i,\eta_i>0;\;\rho_i, \eta_i\in(0,1];\;l_1,l_2,l_4,l_5,l_7,l_8\in(0,1];\;l_3,l_6\in(1,2]$ such that the attack-resilient distributed consensus problem is solved.    
\end{theorem}

\begin{proof}
The proof consists of three parts. In the Part-I, the boundedness of variables, $\boldsymbol{u}_i(k)$, $\boldsymbol{y}_i(k)$, $\hat{\Phi}_i(k)$, $\hat{\boldsymbol{\chi}}_i(k)$, $\hat{\boldsymbol{\theta}}_i(k)$, $\hat{\boldsymbol{d}}_i(k)$, and $\hat{\boldsymbol{\delta}}_{i}(k)$ for $k \in[0, T]$ where $T$ is a finite time instant is illustrated. In Part-II, the boundedness of the tracking error $\tilde{\Phi}_{i}(k)$ is illustrated. In Part-III, the convergence analysis of the $\tilde{\boldsymbol{\chi}}_i(t)$ is illustrated. Then the convergence analysis of the $\hat{\boldsymbol{\chi}}_i(t)$ is given, and hence, the convergence of $\boldsymbol{\chi}_i(t)$ is proved.

Part-I: When $k=0$, the initial values $\|\boldsymbol{u}_i(0)\|<\bar{b}_{\boldsymbol{u}_i(0)},\|\boldsymbol{y}_i(0)\|<\bar{b}_{\boldsymbol{y}_i(0)}$, $\|\hat{\Phi}_i(0)\|<\bar{b}_{\hat{\Phi}_i(0)},\|\hat{\boldsymbol{\chi}}_i(0)\|<\bar{b}_{\hat{\boldsymbol{\chi}}_i(0)}$, $\|\hat{\boldsymbol{\theta}}_i(0)\|<\bar{b}_{\hat{\boldsymbol{\theta}}_i(0)}$, $\|\hat{\boldsymbol{d}}_i(0)\|<\bar{b}_{\hat{\boldsymbol{d}}_i(0)}$, and $\|\hat{\boldsymbol{\delta}}_{i}(0)\|<\bar{b}_{\hat{\boldsymbol{\delta}}_{i}(0)}$ are all given bounded. Then based on the proposed method from Eq. \eqref{eq6} to \eqref{eq12}, one has $\boldsymbol{u}_i(1)$, $\boldsymbol{y}_i(1)$, $\hat{\Phi}_i(1)$, $\hat{\boldsymbol{\chi}}_i(1)$, $\hat{\boldsymbol{\theta}}_i(1)$, $\hat{\boldsymbol{d}}_i(1)$, and $\hat{\boldsymbol{\delta}}_{i}(1)$ are all bounded, and hence one has $\|\boldsymbol{u}_i(1)\|<\bar{b}_{\boldsymbol{u}_i(1)},\|\boldsymbol{y}_i(1)\|<\bar{b}_{\boldsymbol{y}_i(1)}$, $\|\hat{\Phi}_i(1)\|<\bar{b}_{\hat{\Phi}_i(1)},\|\hat{\boldsymbol{\chi}}_i(1)\|<\bar{b}_{\hat{\boldsymbol{\chi}}_i(1)}$, $\|\hat{\boldsymbol{\theta}}_i(1)\|<\bar{b}_{\hat{\boldsymbol{\theta}}_i(1)}$, $\|\hat{\boldsymbol{d}}_i(1)\|<\bar{b}_{\hat{\boldsymbol{d}}_i(1)}$, and $\|\hat{\boldsymbol{\delta}}_{i}(1)\|<\bar{b}_{\hat{\boldsymbol{\delta}}_{i}(1)}$. Consider the boundedness of system variables at any finite time instant $k=T>1$. Following the same steps from $k=0$ to $k=1$, one has $\|\boldsymbol{u}_i(T)\|<\bar{b}_{\boldsymbol{u}_i(T)},\|\boldsymbol{y}_i(T)\|<\bar{b}_{\boldsymbol{y}_i(T)}$, $\|\hat{\Phi}_i(T)\|<\bar{b}_{\hat{\Phi}_i(T)},\|\hat{\boldsymbol{\chi}}_i(T)\|<\bar{b}_{\hat{\boldsymbol{\chi}}_i(T)}$, $\|\hat{\boldsymbol{\theta}}_i(T)\|<\bar{b}_{\hat{\boldsymbol{\theta}}_i(T)}$, $\|\hat{\boldsymbol{d}}_i(T)\|<\bar{b}_{\hat{\boldsymbol{d}}_i(T)}$, and $\|\hat{\boldsymbol{\delta}}_{i}(T)\|<\bar{b}_{\hat{\boldsymbol{\delta}}_{i}(T)}$. Based on the definition of $\boldsymbol{\theta}(k), \max _{k \in[0, T]}|\boldsymbol{\theta}(k)|<\bar{b}_{\boldsymbol{\theta}}$ can be obtained since the boundedness of $u(k)$ and $y(k)$ has been guaranteed for $k \in[0, T]$.

Part-II: Define the estimation error of the PPJM parameter as $\tilde{\Phi}_{i}(k)=\hat{\Phi}_{i}(k)-\Phi_{i}(k)$. Subtracting $\Phi_{i}(k)$ of \eqref{eq10}, one has
\begin{equation}
\label{eq13}
\begin{aligned}
\tilde{\Phi}_i(k)=&\hat{\Phi}_i(k-1) -\Phi_{i}(k)
\\&+\frac{\rho_i  \Phi_i(k-1) \Delta \boldsymbol{u}_i(k-1)\Delta \boldsymbol{u}_i^{\mathrm{T}}(k-1)}{\lambda_i+\|\Delta \boldsymbol{u}_i(k-1)\|^2}
\\&-\frac{\rho_i \hat{\Phi}_i(k-1) \Delta \boldsymbol{u}_i(k-1)\Delta \boldsymbol{u}_i^{\mathrm{T}}(k-1)}{\lambda_i+\|\Delta \boldsymbol{u}_i(k-1)\|^2} 
\\= & \tilde{\Phi}_{i}(k-1)\left(I-\frac{\rho_i  \Delta \boldsymbol{u}_i(k-1)\Delta \boldsymbol{u}_i^{\mathrm{T}}(k-1)}{\lambda_i+\|\Delta \boldsymbol{u}_i(k-1)\|^2}\right)\\&+\Phi_{i}(k-1)-\Phi_{i}(k)
\end{aligned}    
\end{equation}

Taking the norm and the expectation for both sides of \eqref{eq13}, one has
\begin{flalign}
\label{eq14}
\begin{aligned}
\mathbb{E}\left(\left\|\tilde{\Phi}_i(k)\right\|\right)\leq  
&\left\|I-\frac{\rho_i  \Delta \boldsymbol{u}_i(k-1)\Delta \boldsymbol{u}_i^{\mathrm{T}}(k-1)}{\lambda_i+\|\Delta \boldsymbol{u}_i(k-1)\|^2}\right\|
\\&\times\mathbb{E}\left(\left\|\tilde{\Phi}_i(k-1)\right\|\right)+\left\|\Phi_{i}(k-1)-\Phi_{i}(k)\right\|
\end{aligned}&&\raisetag{2\baselineskip}
\end{flalign}

Notice that $\rho_i \Delta \boldsymbol{u}_i(k-1)\Delta \boldsymbol{u}_i^{\mathrm{T}}(k-1)/(\lambda_i+\|\Delta \boldsymbol{u}_i(k-1)\|^2)$ is monotonically increasing for the variable $\|\Delta \boldsymbol{u}_i(k-1)\|^2$, and its minimum value is $\rho_i \epsilon_2^{2} /\left(\lambda_i+\epsilon_2^{2}\right)$ according to \eqref{eq11}. Then, when $\rho_i \in(0,1]$ and $\lambda_i>0$, the following inequality holds $\left\|I-\frac{\rho_i  \Delta \boldsymbol{u}_i(k-1)\Delta \boldsymbol{u}_i^{\mathrm{T}}(k-1)}{\lambda_i+\|\Delta \boldsymbol{u}_i(k-1)\|^2}\right\|\leq \left\|I-\frac{\rho_i \epsilon_2^{2}}{\lambda_i+\epsilon_2^{2}}\right\|\triangleq \varepsilon_1$, where $\varepsilon_1<1$. According to Theorem \ref{thm:2}, one has $\left\|\Phi_{i}(k)\right\| \leq \Phi_i$. Thus, $\left\|\Phi_{i}(k)-\Phi_{i}(k-1)\right\| \leq 2 \Phi_i$. 
Hence, based on the above analysis, using \eqref{eq14}, one has
\begin{equation}
\label{eq15}
\begin{aligned}
\mathbb{E}\left\{\left\|\tilde{\Phi}_{i}(k)\right\|\right\} & \leq \varepsilon_1 \mathbb{E}\left\{\left\|\tilde{\Phi}_{i}(k-1)\right\|\right\}+2 \Phi_i\\
& \leq \varepsilon_1^{2} \mathbb{E}\left\{\left\|\tilde{\Phi}_{i}(k-2)\right\|\right\}+2\varepsilon_1 \Phi_i+2 \Phi_i\\
& \leq \cdots \\
& \leq \varepsilon_1^{k-1} \mathbb{E}\left\{\left\|\tilde{\Phi}_{i}(1)\right\|\right\}+\frac{2 \Phi_i\left(1-\varepsilon_1^{k-1}\right)}{1-\varepsilon_1}
\end{aligned}
\end{equation}which implies that $\tilde{\Phi}_{i}(k)$ is bounded in the mean square sense. Due to the fact that $\Phi_{i}(k)$ is bounded, $\hat{\Phi}_{i}(k)$ also is bounded.

Part-III: Define the estimation error of the observer as $\tilde{\boldsymbol{\chi}}_{i}(k)=\boldsymbol{\chi}_{i}(k)-\hat{\boldsymbol{\chi}}_{i}(k)$, the dynamics of observer error can be obtained by making a difference from \eqref{eq5} and \eqref{eq6}
\begin{equation}
\label{eq16}
\begin{aligned}
\tilde{\boldsymbol{\chi}}_i(k+1)= &(1-l_3)\tilde{\boldsymbol{\chi}}_i(k)-\left(\tilde{\Phi}_i(k)\Delta \boldsymbol{u}_i(k)
+ \tilde{\boldsymbol{\theta}}_i(k)\right)\\&-l_1\hat{\boldsymbol{\delta}}_{i}(k)-l_2\hat{\boldsymbol{d}}_i(k) 
\end{aligned}
\end{equation}

Since the boundedness of $\Delta \boldsymbol{u}_i(k)$, $\tilde{\Phi}_{i}(k)$, $\tilde{\boldsymbol{\theta}}_i(k)$, $\hat{\boldsymbol{d}}_i(k)$, and $\hat{\boldsymbol{\delta}}_{i}(k)$ have been proved, one has $\mathbb{E}\left(\left\|\left(\tilde{\Phi}_i(k)\Delta \boldsymbol{u}_i(k)
+ \tilde{\boldsymbol{\theta}}_i(k)\right)-l_1\hat{\boldsymbol{\delta}}_{i}(k)-l_2\hat{\boldsymbol{d}}_i(k)\right\|\right)\leq\psi_1$. By selecting observer gain $l_3\in(0,2)$, one has
\begin{flalign}
\label{eq17}
\begin{aligned}
\mathbb{E}\left(\left\|\tilde{\boldsymbol{\chi}}_i(k+1)\right\|\right)\leq& (1-l_3)\mathbb{E}\|\tilde{\boldsymbol{\chi}}_i(k)\|+\psi_1 \\
 \leq &(1-l_3)^{2}\mathbb{E}\|\tilde{\boldsymbol{\chi}}_i(k-1)\|+(1-l_3) \psi_1+\psi_1 
\\ \leq &\cdots  
\\ \leq &(1-l_3)^{k}\|\tilde{\boldsymbol{\chi}}_i(1)\|+\frac{\psi_1\left[1-(1-l_3)^{k-1}\right]}{1-(1-l_3)} .
\end{aligned}&&\raisetag{2.7\baselineskip}
\end{flalign}

Hence, the bounded convergence of the observer error $\tilde{\boldsymbol{\chi}}_i(k)$ in the mean square sense is bounded.

Substitute \eqref{eq12} into \eqref{eq6}, one has
\begin{flalign}
\label{eq18}
\begin{aligned}
\hat{\boldsymbol{\chi}}_i(k+1)
= &\hat{\boldsymbol{\chi}}_i(k)
-\frac{\eta_i\hat{\Phi}_i(k)\hat{\Phi}_i(k)^{\mathrm{T}}\left(\hat{\boldsymbol{\chi}}_i(k)+l_3 \tilde{\boldsymbol{\chi}}_i(k)\right)}{\mu_i+\|\hat{\Phi}_i(k)\|^2}
\\&-\frac{\eta_i\hat{\Phi}_i(k)\hat{\Phi}_i(k)^{\mathrm{T}}\left(l_1\hat{\boldsymbol{\delta}}_{i}(k)+l_2\hat{\boldsymbol{d}}_i(k)\right)}{\mu_i+\|\hat{\Phi}_i(k)\|^2}
\\&+ l_1\hat{\boldsymbol{\theta}}_i(k)+l_1\hat{\boldsymbol{\delta}}_{i}(k)+l_2\hat{\boldsymbol{d}}_i(k)+l_3 \tilde{\boldsymbol{\chi}}_i(k)
\\=&\gamma_i(k)\hat{\boldsymbol{\chi}}_i(k)+\omega_i(k)
\end{aligned}&&\raisetag{3\baselineskip}
\end{flalign}
where $\gamma_i(k)=I-\eta_i\hat{\Phi}_i(k)\hat{\Phi}_i(k)^{\mathrm{T}}/(\mu_i+\|\hat{\Phi}_i(k)\|^2)$ and $\omega_i(k)=(I-\eta_i\hat{\Phi}_i(k)\hat{\Phi}_i(k)^{\mathrm{T}}/(\mu_i+\|\hat{\Phi}_i(k)\|^2))(l_1\hat{\boldsymbol{\delta}}_{i}(k)+l_2\hat{\boldsymbol{d}}_i(k)+l_3\tilde{\boldsymbol{\chi}}_i(k)).$
Taking the norm and the expectation for both sides of \eqref{eq18}, and using Lemma \ref{lem:3}, one has 
\begin{equation}
\label{eq19}
\begin{aligned}
\mathbb{E}\left(\left\|\hat{\boldsymbol{\chi}}_{i}(k+1)\right\|\right) \leq&\left\|\gamma_i(k)\right\|\mathbb{E}\left(\left\|\hat{\boldsymbol{\chi}}_{i}(k)\right\|\right)+\left\|\omega_i(k)\right\| \\\leq&\left\|\gamma_i(k)\right\|_{v}\left\|\hat{\boldsymbol{\chi}}_{i}(k)\right\|+\left\|\omega_i(k)\right\|     
\end{aligned}
\end{equation}
where the norm $\|\cdot\|_{v}$ is selected as a proper induced norm.

According to Lemma \ref{lem:2}, the Gershgorin disk of $\gamma_i(k)$ is \begin{equation}
\label{eq20}
\begin{aligned}
D_{i, k}=\left\{z_{i}|\left|z_{i}-\left|1-\alpha_i(k)\right|\right|\leq\left|\beta_i\right|\right\}, k=\left\{1, \ldots, n\right\}
\end{aligned}    
\end{equation}
where $z_{i}$ is the characteristic root of $\gamma_i(k)$, $\alpha_i(k)=(\rho_i \sum_{k=1}^{n} \hat{\Phi}_{i, p k}(k) \hat{\Phi}_{i, p k}^{\mathrm{T}}(k))/(\mu_i+\|\hat{\Phi}_i(k)\|^2)$, and $\beta_i(k)=(\sum_{l=1, l \neq p}^{n} \rho_i \sum_{k=1}^{n} \hat{\Phi}_{i, p k}(k) \hat{\Phi}_{i, l p}^{\mathrm{T}}(k)) / (\mu_i+\|\hat{\Phi}_i(k)\|^2)$. Using the triangle inequality, \eqref{eq20} is rewritten as $D_{i, k}=\left\{z_{i}||z_{i}|\leq\left|1-\alpha_i(k)\right|+\left|\beta_i\right| \right\}.$ From the proof of Theorem 2 in \cite{hou2019model}, the eigenvalue of $\gamma_i(k)$ satisfies $\left|z_{i}\right|_{\gamma_i(k)}<1$ when $\left\|\Phi_{i}(k)\right\| \leq \Phi_i,\left\|\hat{\Phi}_{i}(k)\right\| \leq \hat{\Phi}_{i}$. There exists a positive constant $\Lambda$ such that $s\left[\gamma_i(k)\right]<1-\rho_i \Lambda<1,$ where $s\left(\gamma_i(k)\right)=\max \left|z_{i}\right|$ is the spectral radius of matrix $\gamma_i(k)$. Based on Lemma \ref{lem:3}, there exists a constant $\kappa>0$ such that\begin{equation}
\label{eq21}
\left\|\gamma_i(k)\right\|_{v} \leq s\left(\gamma_i(k)\right)+\kappa=1-\rho_i \Lambda+\kappa=\varepsilon_2<1 
\end{equation}

Based on previous discussion, the boundedness of $\tilde{\boldsymbol{\chi}}_{i}(k), \Phi_{i}(k-1), \hat{\Phi}_{i}(k), \hat{\boldsymbol{\delta}}_{i}(k), $ and $\hat{\boldsymbol{d}}_{i}(k)$ was guaranteed. The controller cannot change too quickly in practical application \cite{xu2014novel}; hence $\Delta \boldsymbol{u}_i(k-1)$ is bounded. In addition, a bounded change of the input can produce a bounded change of the output in practical. Thus, the boundedness of $\Delta \boldsymbol{\chi}_i(k-1)$ is guaranteed. By the above analysis, the boundedness of $\Delta \boldsymbol{\chi}_i(k)$ is guaranteed. Since $\hat{\boldsymbol{\delta}}_{i}(k)$ and $\tilde{d}_{i}(k)$ are bounded, the boundedness of $\hat{\boldsymbol{d}}_{i}(k)$ is guaranteed. Hence, there exist positive constants $\psi_1$, $\psi_2$, $\psi_3$, and $\psi_4$ such that $\mathbb{E}\left(\left\|l_1\left(\hat{\Phi}_i(k)\Delta \boldsymbol{u}_i(k)
+ \hat{\boldsymbol{\theta}}_i(k)\right)\right\|\right) \leq \psi_1,$ $\mathbb{E}\left(\left\|\hat{\Phi}_{i, 1}(k) \gamma_i(k) \Delta \boldsymbol{\chi}_{i}(k)\right\|\right) \leq \psi_2,$ $\mathbb{E}\left(\left\|\hat{\boldsymbol{d}}_{i}(k)\right\|\right) \leq \psi_3$, and $\mathbb{E}\left(\left\|\hat{\boldsymbol{\delta}}_{i}(k)\right\|\right) \leq \psi_4.$ Thus, there exists a positive constant $\varpi_{i}>0$ satisfying $\mathbb{E}\left(\left\|\omega_i(k)\right\|\right) \leq \varpi_{i}$. Then based on \eqref{eq21}, by using \eqref{eq19}, one has
\begin{equation}
\label{eq22}
\begin{aligned}
\mathbb{E}\left(\left\|\hat{\boldsymbol{\chi}}_i(k+1)\right\|\right)\leq& \varepsilon_2\mathbb{E}\|\hat{\boldsymbol{\chi}}_i(k)\|+\varpi_{i} \\
 \leq &\varepsilon_2^{2}\mathbb{E}\|\hat{\boldsymbol{\chi}}_i(k-1)\|+\varepsilon_2 \varpi_{i}+\varpi_{i} 
\\ \leq &\cdots  
\\ \leq &\varepsilon_2^{k}\|\hat{\boldsymbol{\chi}}_i(1)\|+\frac{\varpi_{i}\left[1-\varepsilon_2^{k-1}\right]}{1-\varepsilon_2} .
\end{aligned}
\end{equation}

According to \eqref{eq17} and \eqref{eq22}, $\hat{\boldsymbol{\chi}}_{i}(k)$ and $\tilde{\boldsymbol{\chi}}_{i}(k)$ are all bounded in the mean square sense. Therefore, the consensus error $\boldsymbol{\chi}_i(k)$ is also bounded in the mean square sense, that is, $\boldsymbol{\chi}_i(k)$ converges to the following set $\left\{\boldsymbol{\chi}_i(k)\left|\mathbb{E}\left[\left(\left\|\boldsymbol{\chi}_i(k)\right\|\right)\right] \leq\mathbb{E}\left[\left(\left\|\hat{\boldsymbol{\chi}}_{i}(k)\right\|+\left\|\tilde{\boldsymbol{\chi}}_{i}(k)\right\|\right)\right] \leq \Upsilon_i\right.\right\}$, where $\Upsilon_i>0$ is the upper bound of the consensus error $\boldsymbol{\chi}_i(k)$. 
The proof is completed.
\end{proof}
\section{Simulation Results}
\label{sec:Simulation Results}
In this section, two simulation case studies are performed to demonstrate
the efficacy of the proposed control algorithms.
\subsection{Leaderless Consensus Control}
In this case, taking $n_{\boldsymbol{y}_i}=n_{\boldsymbol{u}_i}=2$, a MAS is considered as $y_{i,1}(k+1) = \frac{y_{i,1}(k) u_{i,1}(k)}{1 + y_{i,1}(k)^{w_{i,1}}} + w_{i,2} u_{i}(k)$ and $y_{i,2}(k+1) = \frac{y_{i,2}(k) u_{i,2}(k)}{1 + y_{i,1}(k)^{w_{i,3}}+y_{i,2}(k)^{w_{i,4}}} + w_{i,5} u_{i}(k),$ where $i \in \left\{1, 2, 3\right\}$; and $k\in(0, T]$ where $T=1500$; and $w_{i,1}$ are $  w_{i,2}$ are positive constants. The communication graph is considered as in Fig. \ref{figures/communicationgraph1.pdf}. The FDI attacks are considered as
$\boldsymbol{\delta}_{i,1}(k, \boldsymbol{y}_i(k))= 0.5 \cos\left(5\pi k / T\right) \sin(y_{i,1}(k))+ 0.5 \cos\left(4\pi k / T\right) \sin(y_{i,1}(k)) \cos(y_{i,2}(k))$ and $\boldsymbol{\delta}_{i,2}(k, \boldsymbol{y}_i(k))= 0.5 \cos\left(5\pi k / T \right) \cos(y_{i,1}(k))
+ 0.5 \sin\left(2\pi k / T \right) \sin(y_{i,1}(k)) \cos(y_{i,2}(k))$. The FDI attack signals are shown in Fig. \ref{figures/scenario1_FDIA.pdf}.
\vspace{-18pt} 
\begin{figure}[H]
    \centering
    \begin{subfigure}[b]{0.06\textwidth}
        \includegraphics[width=\textwidth]{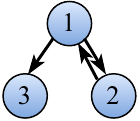}
        \caption{}
        \label{figures/communicationgraph1.pdf}
    \end{subfigure}
    \hspace{0.04\textwidth} 
    \begin{subfigure}[b]{0.1\textwidth}
        \includegraphics[width=\textwidth]{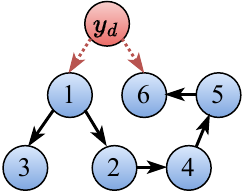}
        \caption{}
        \label{figures/communicationgraph2.pdf}
    \end{subfigure}
    \vspace{-5pt} \caption{Communication graph. (a) Leaderless consensus control scenario; (b) Leader-follower tracking control scenario.}
    \label{fig:scope}
\end{figure}
\vspace{-15pt} 
\begin{figure}[H]
\centering
\includegraphics[width=0.41\textwidth]{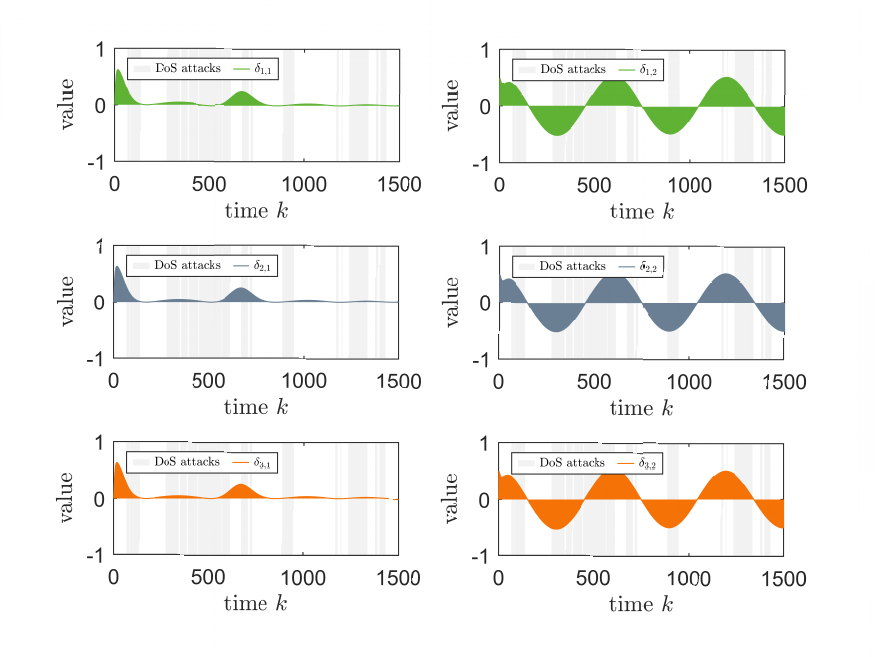}
\caption{FDI and DoS attacks.}
\label{figures/scenario1_FDIA.pdf}
\end{figure}
\vspace{-11pt} 
The disturbances are considered as $d_{i,1} = 0.1 \cos(2\pi k/T)
$ and $d_{i,2} = 0.1 \sin(2\pi k/T)$. The control parameters are selected as \(\eta_{i}=0.1\), \(\rho_{i}=0.1\), \(\lambda_i=1\), and \(\mu_i=1\).  The simulation results are depicted in Fig. \ref{figures/scenario1_y_performance.pdf}, which indicates that the followers, under the proposed framework, achieve resilient performance with an immediate response speed, regardless of FDI and DoS attacks.
\vspace{-5pt} 
\begin{figure}[H]
    \centering
    \begin{subfigure}[b]{0.41\textwidth}
        \centering
        \includegraphics[width=\textwidth]{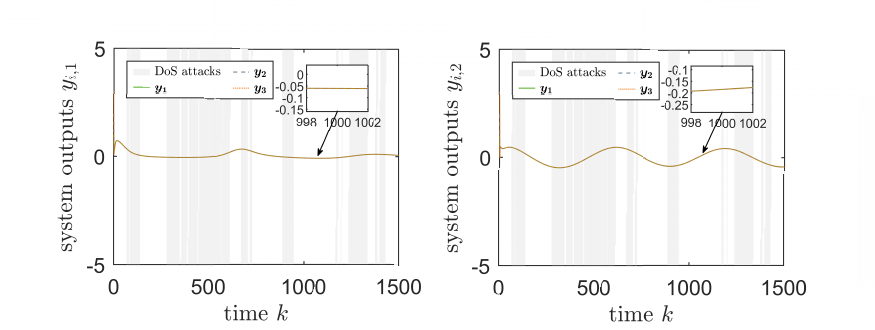}

        \label{fig:scenario1_y_performance}
    \end{subfigure}
    \begin{subfigure}[b]{0.41\textwidth}
        \centering
        \includegraphics[width=\textwidth]{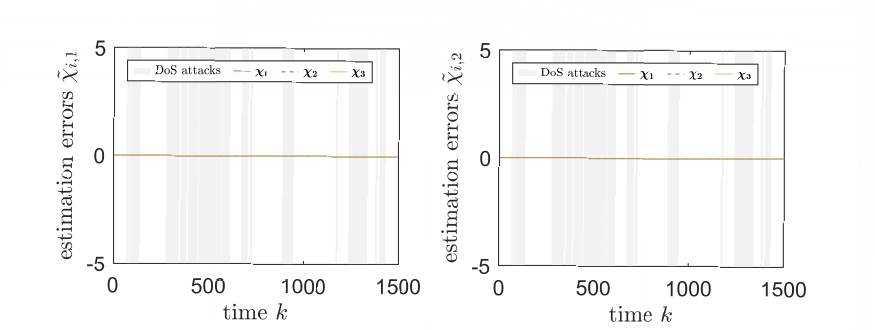}

        \label{fig:scenario1_estimation_errors}
    \end{subfigure}
    \vspace{-5pt}
    \caption{Performance of $y_{i}$ (top); and Estimation errors $\tilde{\boldsymbol{\chi}}_i$ using the proposed method under attacks (bottom).}
    \label{figures/scenario1_y_performance.pdf}
\end{figure}
\subsection{Leader-follower Tracking Control}
In this case, taking $n_{\boldsymbol{y}_i}=n_{\boldsymbol{u}_i}=2$, a MAS is considered as $y_{i,1}(k+1) = \frac{y_{i,1}(k) u_{i,1}(k)}{1 + y_{i,1}(k)^{w_{i,1}}} + w_{i,2} u_{i}(k)$ and $y_{i,2}(k+1) = \frac{y_{i,2}(k) u_{i,2}(k)}{1 + y_{i,1}(k)^{w_{i,3}}+y_{i,2}(k)^{w_{i,4}}} + w_{i,5} u_{i}(k),$ where $i \in \left\{1, \cdots, 6\right\}$; and $k\in(0, T]$ where $T=1500$; and $w_{i,1}$ are $  w_{i,2}$ are positive constants. The desired \(\boldsymbol{y}_0(k)\) is considered as $\boldsymbol{y}_0(k) =\begin{bmatrix} 5 & 2 \end{bmatrix}^{\mathrm{T}},0 \leq k < 500;\begin{bmatrix} 2 & 4 \end{bmatrix}^{\mathrm{T}},500 \leq k < 1000;\begin{bmatrix} 4 & 3 \end{bmatrix}^{\mathrm{T}},1000 \leq k < 1500.$ The communication graph is considered as in Fig. \ref{figures/communicationgraph2.pdf}.
The FDI attacks are considered as
$\boldsymbol{\delta}_{i,1}(k, \boldsymbol{y}_i(k))= 0.5 \cos\left(5\pi k / T\right) \sin(y_{i,1}(k))+ 0.5 \cos\left(4\pi k / T\right) \sin(y_{i,1}(k)) \cos(y_{i,2}(k))$ and $\boldsymbol{\delta}_{i,2}(k, \boldsymbol{y}_i(k))= 0.5 \cos\left(5\pi k / T \right) \cos(y_{i,1}(k))
+ 0.5 \sin\left(2\pi k / T \right) \sin(y_{i,1}(k)) \cos(y_{i,2}(k))$. The FDI attack signals are shown in Fig. \ref{figures/scenario2_FDI.pdf}.  

\begin{figure}[H]
\centering

\includegraphics[width=0.42\textwidth]{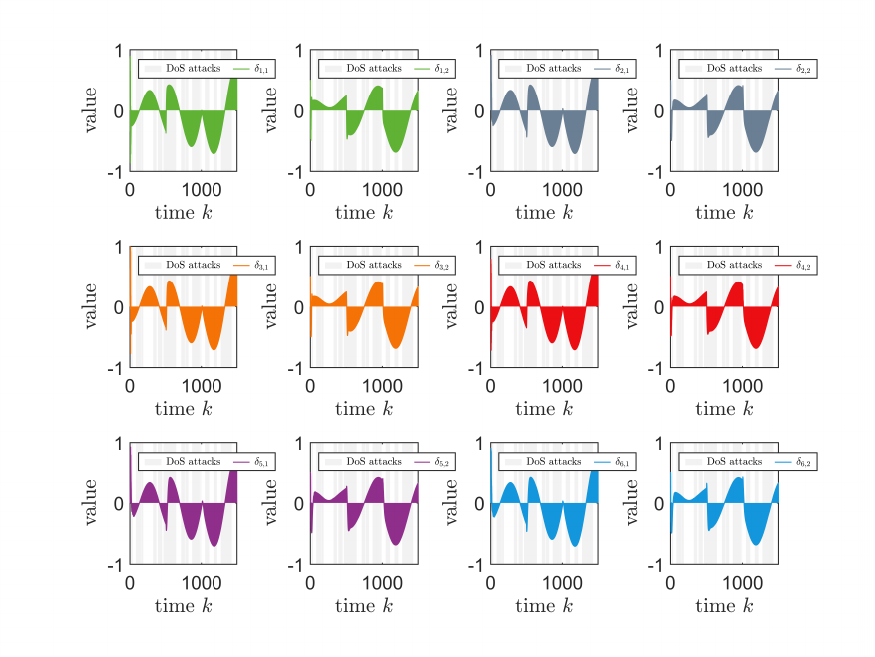}
\caption{FDI and DoS attacks.
}
\label{figures/scenario2_FDI.pdf}
\end{figure}

The disturbances are considered as $d_{i,1} = 0.1 \cos(2\pi k/T)
$ and $d_{i,2} = 0.1 \sin(2\pi k/T)$. The parameters are selected as \(\eta_{i}=0.1\), \(\rho_{i}=0.1\), \(\lambda_i=1\), and \(\mu_i=1\). The simulation results are shown in Fig. \ref{figures/scenario2_y_performance.pdf}. It is evident that the followers can track the leader’s signals within a small boundary, and the tracking errors of the distributed outputs also remain within a small boundary. This indicates that bounded tracking errors are guaranteed even under FDI and DoS attacks.

\begin{figure}[H]
    \centering
    \begin{subfigure}[b]{0.42\textwidth}
        \centering
        \includegraphics[width=\textwidth]{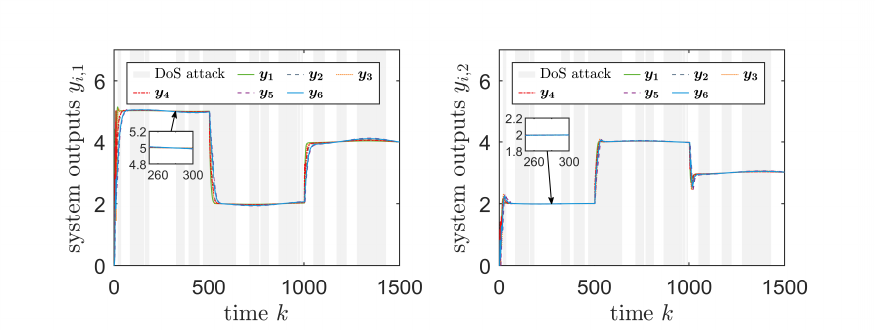}

        \label{fig:scenario2_y_performance}
    \end{subfigure}
    \begin{subfigure}[b]{0.42\textwidth}
        \centering
        \includegraphics[width=\textwidth]{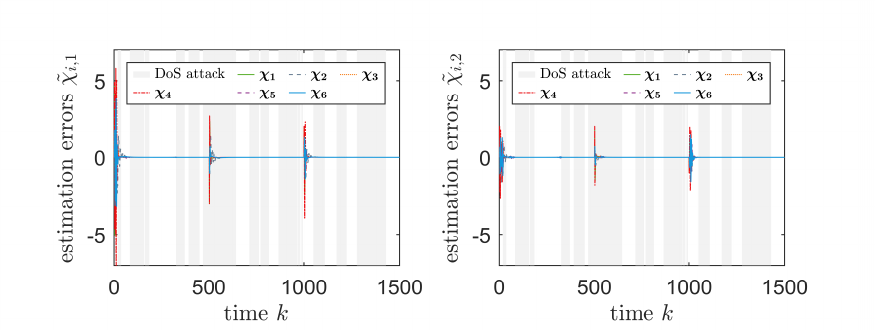}

        \label{fig:scenario2_estimation_errors}
    \end{subfigure}
    \caption{Performance of $y_{i}$ (top); and Estimation errors $\tilde{\boldsymbol{\chi}}_i$ using the proposed method under attacks (bottom).}
    \label{figures/scenario2_y_performance.pdf}
\end{figure}
The simulation results using the conventional MFAC method under attacks are depicted in Fig. \ref{figures/pure_MFAC_DoS_FDI.pdf}, demonstrating that the proposed method achieves significantly better performance under FDI and DoS attacks compared to the conventional MFAC method.
\begin{figure}[H]
\centering
\includegraphics[width=0.42\textwidth]{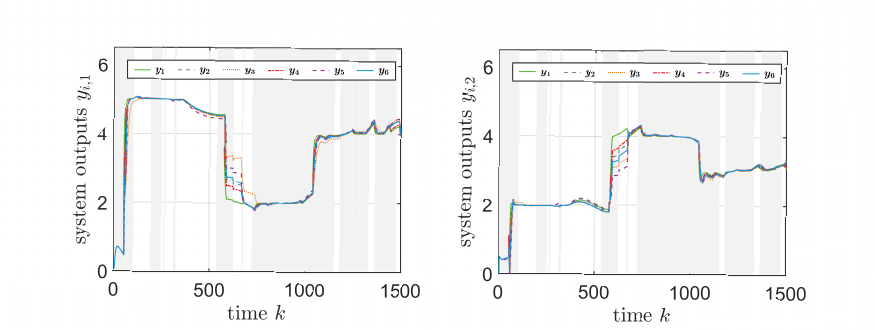}
\caption{Performance of $y_{i}$ using the conventional MFAC.}
\label{figures/pure_MFAC_DoS_FDI.pdf}
\end{figure}

\section{CONCLUSION}
\label{sec:Conclution}

This letter has investigated the attack-resilient consensus control problem under FDI and DoS attacks. A novel attack-resilient observer-based MFAC scheme has been designed to estimate FDI attacks, external disturbances, and lumped disturbances, as well as to compensate for DoS attacks. A rigorous stability analysis has been provided to ensure the boundedness of the distributed neighborhood estimation consensus error. The effectiveness of the proposed approach has been validated through numerical examples involving both leaderless consensus and leader-follower consensus. Compared to the conventional MFAC scheme, it has demonstrated significantly improved attack-resilient performance over existing data-driven control approaches.

\vfill
\bibliographystyle{IEEEtran}
\bibliography{bare_jrnl_new_sample4}
\end{document}